\documentclass
[aps,preprint,preprintnumbers,amsmath,amssymb,nofootinbib]{revtex4}%
\usepackage{amssymb}
\usepackage{amsmath}
\usepackage{bm}
\usepackage{amsthm}
\usepackage{color}
\usepackage{amsfonts}
\usepackage{graphicx}%
\providecommand{\U}[1]{\protect\rule{.1in}{.1in}}
\theoremstyle{plain}
\newtheorem{thm}{Theorem}
\theoremstyle{definition}

\theoremstyle{proposition}
\newtheorem{prop}{Proposition}
\theoremstyle{lemma}

\theoremstyle{corollary}

\begin{document}
\title{On the proof of the Thin Sandwich Conjecture in arbitrary dimensions.}
\author{R. Avalos, F. Dahia, C. Romero and J. H. Lira.}

\begin{abstract}
In this paper we show the validity, under certain geometric conditions, of
Wheeler's thin sandwich conjecture for higher dimensional theories of gravity.
We extend the results shown by R. Bartnik and G. Fodor for the 3-dimensional
case in \cite{Bartnik} in two ways. On the one hand we show that the results
presented in \cite{Bartnik} are valid in arbitrary dimensions, and on the
other hand we show that the geometric hypotheses needed for the proofs can
always be satisfied, which constitutes in itself a new result for the
3-dimensional case. In this way, we show that on any compact $n$-dimensional
manifold, $n\geq3$, there is an open set in the space of all possible initial
data where the thin sandwich problem is well-posed.

\end{abstract}
\affiliation{Departamento de F\'{\i}sica, Universidade Federal da Para\'{\i}ba, Caixa
Postal 5008, 58059-970 Jo\~{a}o Pessoa, PB, Brazil.}
\affiliation{E-mail: rodrigo.avalos@fisica.ufpb.br; fdahia@fisica.ufpb.br;
cromero@fisica.ufpb.br; jorge.lira@mat.ufc.br}
\maketitle

\section{Introduction}

As is well known, the Cauchy problem for general relativity consists in
finding a solution of the Einstein equations in a 4-dimensional Lorentzian
manifold, which satisfies some prescribed initial conditions on a
3-dimensional Riemannian hypersurface. This can be understood as studying
whether we can propagate some initial space-like hypersurface, such that the
resulting space-time satisfies the Einstein equations. This problem has long
been studied and there are results which show that general relativity has a
well-posed Cauchy problem for initial data satisfying some constraint
equations \cite{Geroch-CB}. A detailed review on this topic can be found in
\cite{C-B1} and \cite{Ringstrom}. These contraint equations imply that we
cannot arbitrarily give the initial data set for the Cauchy problem,
motivating the study of these equations so as to determine under what
conditions it has a solution, and what part of this data can in fact be given
arbitrarily on the initial manifold. It is customary to regard this system as
a system of partial differential equations (PDE) for a Riemannian metric and for
some $(0,2)$ symmetric tensor field defined on this Riemannian hypersurface,
which in the end will play the role of the extrinsic curvature. We now know
that under some hypotheses on the topology of the space-like manifold, we can
specify a conformal metric to the physical one (that is the one which will
solve the constraints and hence will have a development in space-time) and the
trace of the second fundamental form, and then get a well-posed system for the
remaining undetermined quantities \cite{C-B1}. Another way to look at this
problem was proposed by Wheeler. His idea was to consider space-time as a
curve in what he called Superspace. Intuitively, given a 3-dimensional
manifold $M$, the Superspace $S(M)$ related to it would be the space of
\textit{geometries} that can be defined on $M$. In this way, a point in $S(M)$
is regarded as an equivalence class $(M,[g])$, represented by a pair $(M,g)$
with $g$ a Riemannian metric defined on $M$, where two Riemannian metrics are
considered equivalent if they are related to each other by a diffeomorphism
via pullback. A detailed review on this structure can be found in
\cite{Giulini1}. With this in mind we can think of space-time as a curve in
Superspace. In this context Wheeler proposed the Thin Sandwich Problem (TSP)
\cite{Wheeler}, where the idea is to give as initial data a Riemannian metric
$g$ and a tangent vector $\partial_{t}g$ to $(M,g)$, and then study whether we
can solve the constraint equations for these initial data. If we can, then
these initial data has a unique Cauchy development in space-time. This means
that there would be a unique curve in Superspace satisfying these initial
conditions and compatible with the Einstein equations. This problem has been
recently investigated by some authors \cite{Bartnik},\cite{B-O}%
,\cite{Giulini2}. In this paper, we will be particularly interested in the
results obtained by Bartnik and Fodor, who, for the 3-dimensional case, which
is the arena of classical general relativity, have established sufficient
conditions for the TSP to be well-posed \cite{Bartnik}. {More precisely, they
have shown that given some \textit{free} data $(g_{0},\dot{g}_{0},\epsilon
_{0},S_{0})$ satisfying some specific geometric conditions and for which a
solution of the constraints in their thin-sandwich formulation exists, there
is a neighbourhood of this free data set where the TSP is well-posed. Even
though this result mainly relies on both elliptic theory and an implicit
function argument, which do not generally depend on the dimension, in the
proofs they explicitly take advantage of the fact that they are working in
3-dimensions to manipulate expressions in a way which is not practical in
arbitrary dimensions. But since, just as for the evolution problem in GR, the
setting of the constraint equations in its classical formulation does not
strongly depend on the dimension, it would be expected that the results
presented in \cite{Bartnik} should extend to arbitrary dimensions ($n\geq3$).
We will show that this is actually true, and that there is in fact an
$n$-dimensional analog of the Bartnik-Fodor theorem. Also, in \cite{Bartnik},
in order to show that there are reference solutions for the constraint
equations where their main theorem applies, they produce an example using the
initial data induced by the spatially compactified Friedman-Robertson-Walker
cosmological solution with $k=-1$, where all the conditions needed for this
theorem are satisfied. Nevertheless, it is not shown that on any compact
3-dimensional manifold a reference solution exists. In this paper we show that
this last statement actually holds, that is, on any compact $n$-dimensional
manifold, $n\geq3$, there are reference solutions of the constraint equations
satisfying all the hypotheses needed to apply the implicit function argument.
In this way we will be concerned with the local well-posedness of the TPS,
where by this we mean that we will show that in a neighbourhood of free data
with specific properties the TSP has a unique solution. It should be stressed
that we do not expect this problem to be well-posed for arbitrary data. For
instance, following an argument presented by Belasco and Ohanian in
\cite{B-O}, if we choose data $(g,\dot{g},\epsilon,S)$ such that
$2\epsilon-R(g)>0$ and $\dot{g}=\pounds _{X}g$ for any smooth vector field
$X$, then no solution for the the TSP can exist on a compact (without
boundary) and connected manifold $M$. }

We would also like to draw the reader's attention to \cite{Giulini2}, where a
generalization of \cite{Bartnik} is presented which includes models for the
matter fields in a more realistic way. Even though we did not follow this
approach, it is worth to emphasize that the framework and techniques presented
in this paper could provide interesting future developments on
higher-dimensional TSP as well.

\section{Statement of the problem.}

The Cauchy problem for general relativity (GR) consists in the following.
Given an initial data set $(M,g,K)$ where $M$ is an $n$-dimensional smooth
Riemannian manifold with metric $g$ and $K$ is a symmetric second rank tensor
field, a development of this initial data set is a space-time $(V,\bar g)$,
such that there exists an embedding into $\iota:M\mapsto V$ with the following
properties:\newline i) The metric $g$ is the pullback of $\bar g$ by the
embedding $\iota$, that is $\iota^{*}\bar g=g$.\newline ii) The image by
$\iota$ of $K$ is the second fundamental form of $\iota(M)$ as a submanifold
of $(V,g)$.

In the Cauchy problem for GR we look for a development of an initial data set
such that the resulting space-time satisfies the Einstein equations. It is
customarily assumed that $V=M\times\mathbb{R}$. Since this is a consequence of
global hyperbolicity we do not regard it as a physical obstruction, and thus
we will adopt this usual setting.

At this point, to study the Cauchy problem, it is customary to consider an
$(n+1)$-dimensional space-time $(V,\bar g)$ and then make an ``$(n+1)$%
-splitting" for the metric $\bar g$. This means that we consider local
co-frames where we can write the metric $\bar g$ in a convenient way, such
that we have a ``space-time splitting". In order to do this, a vector field
$\beta$, which is constructed so as to be tangent to each hypersurface
$M\times\{t\}$, is used to define the following local frame
\begin{align*}
e_{i}  &  = \partial_{i}\; ,\; \; i=1,\cdots,n\\
e_{0}  &  = \partial_{t} - \beta
\end{align*}
and its dual coframe
\begin{align*}
\theta^{i}  &  = dx^{i} +\beta^{i}dt \; , \; \; i=1,\cdots,n\\
\theta^{0}  &  = dt
\end{align*}
Then we can write the metric $\bar g$ in the following way
\begin{align*}
\bar g = -N^{2}\theta^0\otimes\theta^0+ g_{ij}\theta^{i}\otimes\theta^{j}%
\end{align*}
where the function $N$ is a positive function referred to as the
\textit{lapse} function, while the vector field $\beta$ is called the
\textit{shift} vector.

In this adapted frame, the second fundamental form on each $M\times\{t\}$
takes the form
\begin{equation}
\label{curvext}K_{ij}=\frac{1}{2N}(\partial_{t} g_{ij}-(\nabla_{i}\beta
_{j}+\nabla_{j}\beta_{i}))
\end{equation}
where $\nabla$ denotes the induced connection in $M$ compatible with the
induced metric $g$.

As we have already noted, the possibility of finding an Einstenian development
of an initial data set depends on whether the following set of constraint
equations are satisfied by these initial data:
\begin{align}
R_{g} - \vert K\vert^{2}_{g}+(\mathrm{tr}_{g}K)^{2}  &  = 2\epsilon\\
\mathrm{div} K - \nabla\mathrm{tr}_{g}K  &  = S
\end{align}
where $(\epsilon,S)$ denote the induced energy and momentum densities on $M$,
respectively, $R_{g}$ represents the scalar curvature of $g$, $\vert\cdot
\vert_{g}$ denotes the pointwise-tensor norm in the metric $g$ and
$\mathrm{div} K$ denotes the divergence of $K$. These constraint equations are
posed on an $n$-dimensional manifold $M$ and are imposed by the $(n+1)$%
-dimensional Einstein equations (see, for instance, \cite{C-B1}). In
coordinates, these equations become:
\begin{align}
R_{g} - K^{ij}K_{ij} + (K^{i}_{i})^{2}  &  =2\epsilon\label{hamit}\\
\nabla_{j} K^{j}_{i}-\nabla_{i} K^{j}_{j}  &  =S_{i} \label{momentum}%
\end{align}
These equations are considered on a particular initial hypersurface $M\cong
M\times\{t\}$, for example, in the hypersurface defined by $t=0$. If our
initial data set $(M, g, K)$ satisfies these constraints, then, for many
sources of interest, it can be shown that there is an Einstenian development
in our space-time $V$ \cite{C-B1}.

Equations (\ref{hamit})-(\ref{momentum}) are generally posed as a set of
equations for $g$ and $K$. In the context of Wheeler's TSP these equations are
looked as equations for $N$ and $\beta$. In order to do this, we use
(\ref{curvext}) to express (\ref{hamit})-(\ref{momentum}) in terms of the
lapse and shift, and then look for solutions with some prescribed data
$(g,\dot g,\epsilon,S)$, where $\dot g=\partial_{t} g$.

In the scenario of the TSP, suppose that, given some prescribed data $(g,\dot
g,\epsilon, S)$, we have a solution $(N,\beta)$ for the constraint equations.
Furthermore, suppose this solution satisfies $2\epsilon-R_{g}\neq0$ over all
$M$. Then, introducing (\ref{curvext}) in (\ref{hamit}) we can equate the
lapse function in terms of the shift vector and the prescribed data. After
doing this we obtain%

\begin{equation}
\label{lapse}N=\sqrt{\frac{(\mathrm{tr}_{g} \gamma)^{2}-|\gamma|^{2}_{g}%
}{2\epsilon-R_{g}}}%
\end{equation}
where the tensor $\gamma$ has components
\begin{equation}
\gamma_{ij}=\frac{1}{2}\big(\dot g_{ij}-(\nabla_{i}\beta_{j}+\nabla_{j}%
\beta_{i})\big).
\end{equation}

It should be noted that in (\ref{lapse}) we have chosen the positive sign for
the square root, since this choice, which corresponds to the choice of
positive lapse, is related to the choice of a space-time foliation which
\textit{evolves to the future}, whereas the negative sign would represent a
foliation \textit{evolving to the past}. Furthermore, a few comments on the
individual signs of the numerator and denominator are in order. First of all,
note that if $M$ is connected, then the condition $2\epsilon-R_{g}\neq0$ at
each point of $M$ implies that $2\epsilon-R_{g}$ has a definite sign on  $M$.
It should be noted that later on we will impose the condition $2\epsilon
-R_{g}>0$ for a reference solution of the constraint equations, in a
neighbourhood of which we will study the TSP. This condition, which imposes an
energy constraint, forces the numerator in (\ref{lapse}) to be strictly
positive, and, furthermore, requires that $\mathrm{tr_{g}}\gamma\neq0$
$\forall$ $p\in M$. Using (\ref{curvext}) and the definition of $\gamma$, we
see that, if $M$ is connected, this implies that $\tau\doteq\mathrm{tr_{g}K}$
has a definite sign all over $M$. That is, if this initial data set has an
embedding into a space-time satisfying the Einstein equations, then the
hypersurface $M\times\{0\}\cong M$ will be an embedded hypersurface whose mean
curvature has a definite sign. This fact carries a clear physical
interpretation: the sign of the mean curvature is related to whether the
future pointing unit normals are diverging from the hypersurface or
converging, representing, respectively, an expanding or contracting space
evolving in space-time.

Now, replacing (\ref{lapse}) in (\ref{momentum}) shows that the shift vector
satisfies the following equation
\begin{equation}
\label{RTS}\nabla_{i}\Bigg(\sqrt{\frac{2\epsilon-R_{g}}{(\mathrm{tr}_{g}
\gamma)^{2}-|\gamma|^{2}_{g}}}\,\big(\gamma^{i}_{j}-\delta^{i}_{j}%
\mathrm{tr}_{g} \gamma\big)\Bigg)=S_{j},
\end{equation}
that is,
\begin{equation}
\label{div-S}\mathrm{div} \Bigg(\sqrt{\frac{2\epsilon-R_{g}}{(\mathrm{tr}_{g}
\gamma)^{2}-|\gamma|^{2}_{g}}}\,\big(\gamma-\mathrm{tr}_{g} \gamma\,
g\big)\Bigg) = S.
\end{equation}

We have a converse procedure to the one just described. That is, if, for a
given initial data set $(g,\dot g,\epsilon, S)$, (\ref{RTS}) is well-posed and
$\beta$ is a solution of (\ref{RTS}), then taking (\ref{lapse}) as a
definition, the lapse will satisfy (\ref{hamit}).

It is worth to point out that the equations (\ref{div-S}) have a variational
origin (see \cite{Bartnik} and \cite{B-O}). In particular, the first detailed
treatment of the thin-sandwich problem was made using this variational
formulation \cite{B-O}. There, some uniqueness and non-existence results were
shown, including a global uniqueness result (see also \cite{Giulini2}).

We can now state the problem we want to study here. Given a solution
$(N,\beta)$ of the constraint equations obtained from some given data $(g,
\dot g,\epsilon, S)$, can we obtain a solution of the constraint equations for
data "sufficiently near" of these given data? We will first show that, under
certain hypotheses, this can be answered affirmatively and then that these
hypotheses can always be satisfied by some reference solution on a any compact
$n$-dimensional manifold $\forall$ $n\geq3$. Note that proving that for any
set of initial data sufficiently near to $(g,\dot g,\epsilon, S)$ there is a
unique solution of the constraint equations, also proves that, if the
associated evolution problem is well-posed, then for these data there exists a
unique Cauchy development in space-time, and this, in turn, would prove a
restricted form of Wheeler's thin sandwich conjecture.

Before going further, it would be appropriate to remark that when we say that
the quantity $\dot g$ is a given datum, we mean that some symmetric $(0,2)$
tensor field on $M$ is given, and that with this tensor field we construct $K$
from (\ref{curvext}), taking this tensor field as $\partial_{t}g_{ij}$. Then
if we have a solution for the Cauchy problem, this tensor field will coincide
with $\partial_{t}g_{ij}$ on $M\times\{0\}$.

\section{Main Results}

As we have stated above, we need to study whether any initial data set
$(g,\dot g,\epsilon, S)$ sufficiently near to a reference solution of the
constraint equations also satisfies the constraint equations. In order to do
this, we can concentrate ourselves to answer this question just for the set of
equations (\ref{RTS}). In order to proceed, we will assume $M$ to be compact
(without boundary) and write this set of non-linear PDE for the shift vector
in the following way. Let
\[
H_{s} (T^{p}_{q} (M)), \quad s>\frac{n}{2},\,\, s>2,
\]
be the space of $(p,q)$-tensor fields in $M$ with local components in the Sobolev
space $H_{s}(\Omega)$, where $\Omega$ is an open subset of $\mathbb{R}^{n}$.
Denote
\[
\mathcal{E}_{1}\doteq H_{s+3}(T^{0}_{2}M)\times H_{s+1}(T^{0}_{2}M)\times
H_{s+1}(M)\times H_{s}(T^{0}_{1}M)
\]
which is a Banach space with the norm $\Vert\cdot\Vert_{\mathcal{E}_{1}%
}:\mathcal{E}_{1} \to\mathbb{R}$ given by
\begin{align*}
|| (g,\dot g,\epsilon, S)||_{\mathcal{E}_{1}} = \Vert g\Vert_{H_{s+3}}%
+\Vert\dot g \Vert_{H_{s+1}}+\Vert\epsilon\Vert_{H_{s+1}}+\Vert S\Vert_{H_{s}}%
\end{align*}
and let
\[
\mathcal{E}_{2}\doteq H_{s+2}(T^{1}_{0}M) \,\,\, \mbox{ and } \,\,\,
\mathcal{F}\doteq H_{s}(T^{0}_{1}M).
\]
Now suppose that for given data $\psi_{0}\doteq(g_{0},\dot g_{0},\epsilon
_{0},S_{0})\in\mathcal{E}_{1}$ we have a solution $\beta_{0} \in
\mathcal{E}_{2}$. Then, the continuity of all the maps involved guarantees
that (\ref{RTS}) is well-defined in a neighborhood $\mathcal{U}$ of $(\psi
_{0},\beta_{0})$ in $\mathcal{E}_{1}\times\mathcal{E}_{2}$. With this in mind,
we define the map
\begin{align*}
\Phi:\mathcal{U}\subset\mathcal{E}_{1}\times\mathcal{E}_{2}\to\mathcal{F}%
\end{align*}
given by
\begin{equation}
\label{Phi}\Phi(\psi, \beta)\doteq\mathrm{div} \Bigg(\sqrt{\frac
{2\epsilon-R_{g}}{(\mathrm{tr}_{g} \gamma)^{2}-|\gamma|^{2}_{g}}}%
\,\big(\gamma-\mathrm{tr}_{g} \gamma\, g\big)\Bigg) - S
\end{equation}
where we have denoted $\psi= (g,\dot g,\epsilon, S)$, and we are using $\beta$
to denote the shift. Then (\ref{RTS}) can be written as
\begin{equation}
\label{RTH2}\Phi(\psi,\beta)=0.
\end{equation}

Now our problem reduces to the following: we want to see if there are open
sets $\mathcal{V}\subset\mathcal{E}_{1}$, $\mathcal{W}\subset\mathcal{E}_{2}$,
with $\psi_{0}\in\mathcal{V}$ and $\beta_{0}\in\mathcal{W}$, and a unique map
\[
g:\mathcal{V}\to\mathcal{W}
\]
such that
\[
\Phi(\psi,g(\psi))=0 \,\,\, \mbox{ for all }\,\,\, \psi\in\mathcal{V}.
\]
Notice that, in this case, $\beta=g(\psi)\in\mathcal{W}$ would be the solution
to our problem. In order to address this issue, we intend to use the Implicit
Function Theorem. Hence, we need to show that
\begin{equation}
L = \frac{\delta\Phi}{\delta\beta}\bigg|_{(\psi_{0},\beta_{0})}:\mathcal{E}%
_{2}\to\mathcal{F}%
\end{equation}
is an isomorphism. This is precisely the procedure followed in \cite{Bartnik}
in the 3-dimensional case. We will extend their results for arbitrary
dimensions ($n\geq3$). Using (\ref{Phi}), we compute%

\begin{equation}
\label{LRTS}L \mathcal{Y} =\frac{\delta\Phi}{\delta\beta}\bigg|_{(\psi,\beta
)}= \mathrm{div}\Bigg(\frac{1}{N}\bigg(\mathrm{div}\mathcal{Y} g - {}%
^{S}\nabla\mathcal{Y} - \frac{1}{2\epsilon-R_{g}} \langle\pi, \nabla
\mathcal{Y}\rangle\pi\bigg)\Bigg)
\end{equation}
where $\pi$ is the tensor
\begin{equation}
\label{pi}\pi\doteq\frac{1}{N} (\gamma- \mathrm{tr}_{g} \gamma\, g) = K -
\mathrm{tr}_{g} K g,
\end{equation}
which represents the conjugate momentum to $g$ in the Hamiltonian picture of
GR, and
\[
{}^{S} \nabla_{i} \mathcal{Y}_{j} = \frac{1}{2}\big(\nabla_{i} \mathcal{Y}%
_{j}+\nabla_{j} \mathcal{Y}_{i}\big).
\]
We will study the properties of the linearized operator $L$. First of all, it
is clear that $L$ is a linear second order operator. We now have the following proposition.

\begin{prop}
If $\pi$ is a definite operator all over $M$, then the linear operator $L$ is elliptic.
\end{prop}

\begin{proof}
The first thing we need to do is to compute the symbol of the linear operator $L$.  We easily verify that the symbol of $L$ is given by
\begin{align}\label{symbol}
(\sigma(L)(\xi)\cdot\mathcal{Y})^j = \frac{1}{N}\bigg( \frac{1}{2}\xi^j\xi_k -\frac{1}{2}|\xi|^2_g\delta^j_k - \frac{1}{2\epsilon -R_g} \pi^{ij}\xi_i \pi_{k}^{\ell} \xi_\ell\bigg)\mathcal{Y}^k ,
\end{align}
for all $\xi \in \Gamma(T^*M)$ and $\mathcal{Y}\in \Gamma(TM)$. Hence
\begin{align*}
\langle \sigma(L)(\xi)\cdot \mathcal{Y}, \xi\rangle &= (\sigma(L)(\xi)\cdot\mathcal{Y})^j \xi_j= \frac{1}{N}\bigg( \frac{1}{2}|\xi|^2_g \langle \xi, \mathcal{Y}\rangle -\frac{1}{2}|\xi|^2_g \langle \xi, \mathcal{Y}\rangle - \frac{1}{2\epsilon -R_g} \pi(\xi, \xi) \pi(\xi, \mathcal{Y}) \bigg)\\
\,\, &= - \frac{1}{2\epsilon -R_g} \pi(\xi, \xi) \pi(\xi, \mathcal{Y})
\end{align*}
Suppose $\mathcal{Y}\in T_pM$ such that $\sigma(L)(\xi)\cdot\mathcal{Y}=0$ for some $\xi\neq 0$. Then
\begin{equation}
\label{op2}
\pi(\xi, \xi) \pi(\xi, \mathcal{Y}) =0
%0=\frac{\pi^{hj}\xi_h\xi_j\pi_{mk}y^k\xi^m}{2\epsilon-\overline{R}} \; \; \forall \; \; \xi\neq 0.
\end{equation}
for some $\xi\neq 0$. Since by assumption $\pi$ is definite and hence non-degenerate, this implies that $\pi(\xi,\mathcal{Y})=0$, $\xi\neq 0$. Using this information in (\ref{symbol}), we get that if $\mathcal{Y}$ is in the null space of $\sigma(L)(\xi)$, then $\mathcal{Y}=\langle\frac{\xi}{|\xi|^2_g},\mathcal{Y}\rangle \xi$. This two conditions, together with the fact that $\pi$ is non-degenerate, imply that $\mathcal{Y}=0$. Thus $L$ is elliptic.
\end{proof}

It is interesting to note that the condition on $\pi$ being a definite
operator has one particular consequence with a clear physical interpretation.
Note that $\pi$ being definite imposes a condition on $\mathrm{tr_{g}}K$,
since, using (\ref{pi}), we get that $\mathrm{tr_{g}}K=\frac{1}{1-n}%
\mathrm{tr_{g}\pi}$. Also, note that $\pi$ defines an operator $\pi^{\sharp}$
on vector fields, given in components by $\pi^{\sharp}(X)^{i}\doteq\pi_{j}%
^{i}X^{j}$. Note that the symmetry of $\pi$ shows that $\pi^{\sharp}$ defines
a self adjoint operator (with respect to $g$) on each tangent space. That is,
$\langle v,\pi^{\sharp}(w)\rangle=\langle\pi^{\sharp}(v),w\rangle$ for all
$v,w\in T_{p}M$ and $p\in M$. This means that, at each point, there is a
$g$-orthonormal basis diagonalizing $\pi$. Using such basis in order to
compute $\mathrm{tr_{g}}\pi$, we see that the trace is the sum of the
eigenvalues of $\pi$, and thus, that if $\pi$ is definite, the trace must be either strictly
positive or strictly negative. This implies that if $\pi$ is definite on $M$, then, if $M$ is
connected, $\mathrm{tr_{g}}K$ must have constant sign on $M$, and cannot be
zero. Now, if a given initial data set $(g,K)$ satisfying this condition on
the trace of $K$ has a development into a space-time $V$, then the embedded
hypersurface $M\times\{0\}\cong M$ has mean curvature with a definite sign all
over the hypersurface. This, again, can be interpreted as telling us that the
whole hypersurface is either expanding or contracting in its space-time
evolution (at least for short times).

From now on, we will suppose that $\pi$ gives a definite operator at every
point of $M$ so that the last proposition holds. Having in mind that our aim
is to establish sufficient conditions so that $L$ is an isomorphism, the
ellipticity condition just established shows that what we need to do is to
show the injectivity of both $L$ and its formal adjoint $L^{*}$. A
straightforward computation, using integration by parts, gives us that $L$ is
(formally) self-adjoint. This means that for all smooth vector fields
$\mathcal{Y}, \mathcal{Z}$ the following holds
\begin{align*}
\int_{M} \langle L\mathcal{Y}, \mathcal{Z}\rangle\, \mathrm{d}M_{g_{0}}%
=\int_{M} \langle\mathcal{Y}, L\mathcal{Z}\rangle\, \mathrm{d}M_{g_{0}},
\end{align*}
where $\mathrm{d}M_{g_{0}}$ is the Riemannian volume element in $M$ induced by
the metric $g_{0}$. Thus, if $\pi$ is a definite operator on $M$, then $L$ is
a (formally) self-adjoint elliptic operator, and what we need to establish is
its injectivity, which is the content of the following proposition.

\begin{prop}
Consider a reference solution $(\psi,\beta)$ for the TSP on a compact
n-dimensional manifold $M$ satisfying that: \textrm{i)} $\pi$ is a definite
operator on $M$; \textrm{ii)} $2\epsilon-R_{g}>0$ on $M$; \textrm{iii)} given
a function $\mu$, the equation
\begin{equation}
\label{conf-kill}{}^{S}\nabla\mathcal{Y}=\mu K
\end{equation}
has only the solution $\mathcal{Y}=0$, $\mu=0$. Then $L$ is injective.
\end{prop}

\begin{proof}
Recall that
\begin{equation}
L\mathcal{Y} = {\rm div} \bigg(\frac{1}{N} \bigg({\rm div}\mathcal{Y}\, g - \frac{1}{2}\pounds_{\mathcal{Y}} g - \frac{1}{2\varepsilon-R} \langle \pi, \nabla \mathcal{Y}\rangle \pi\bigg)\bigg)
\end{equation}
Let $\Omega$ be a relatively compact open subset in $M$ and let $\eta\in C^\infty_0(\Omega)$ with $\eta \equiv 1$ in $\Omega'\subset \Omega$. Denoting $\mathcal{Z} =\eta \mathcal{Y}$, one obtains
\begin{align*}
\langle  L\mathcal{Y}, \mathcal{Z}\rangle &= Z^j \nabla_i \bigg(\frac{1}{N} \bigg({\rm div}\mathcal{Y}\, \delta^i_j - \frac{1}{2}(\nabla_j \mathcal{Y}^i+\nabla^i \mathcal{Y}_j)- \frac{1}{2\varepsilon-R} \langle \pi, \nabla \mathcal{Y}\rangle \pi^i_j\bigg)\bigg)\\
&= \nabla_i \bigg(\frac{1}{N} \bigg({\rm div}\mathcal{Y}\, \delta^i_j - \frac{1}{2}(\nabla_j \mathcal{Y}^i+\nabla^i \mathcal{Y}_j)- \frac{1}{2\varepsilon-R} \langle \pi, \nabla \mathcal{Y}\rangle \pi^i_j\bigg) Z^j\bigg)\\
&- \frac{1}{N} \bigg({\rm div}\mathcal{Y}\, \delta^i_j - \frac{1}{2}(\nabla_j \mathcal{Y}^i+\nabla^i \mathcal{Y}_j)- \frac{1}{2\varepsilon-R} \langle \pi, \nabla \mathcal{Y}\rangle \pi^i_j\bigg)\nabla_i Z^j
\end{align*}
which yields
\begin{align*}
\langle  L\mathcal{Y}, \mathcal{Z}\rangle &= {\rm div} \bigg(\frac{1}{N} \bigg( \mathcal{Z}\,{\rm div}\mathcal{Y} -\frac{1}{2} \pounds_{\mathcal{Y}} g (\mathcal{Z}, \cdot) - \frac{1}{2\varepsilon-R} \langle \pi, \nabla \mathcal{Y}\rangle \pi(\mathcal{Z},\cdot) \bigg)\bigg)\\
& -\frac{1}{N} \bigg({\rm div}\mathcal{Y}\, {\rm div} \mathcal{Z}- \langle {}^S\nabla \mathcal{Y}, \nabla \mathcal{Z}\rangle - \frac{1}{2\varepsilon-R} \langle \pi, \nabla \mathcal{Y}\rangle \langle \pi, \nabla\mathcal{Z}\rangle\bigg).
\end{align*}
We conclude that
\begin{align*}
\int_\Omega\langle L  \mathcal{Y}, \mathcal{Z}\rangle\, {\rm d}M &= \int_{\partial\Omega} \frac{1}{N}\bigg (\langle \mathcal{Z}, \nu\rangle \,{\rm div}\mathcal{Y} -\frac{1}{2} \pounds_{\mathcal{Y}} g (\mathcal{Z}, \nu) - \frac{1}{2\varepsilon-R} \langle \pi, \nabla \mathcal{Y}\rangle \langle\pi(\mathcal{Z}), \nu\rangle\bigg)\, {\rm d}\partial M\\
&- \int_\Omega \frac{1}{N} \bigg({\rm div}\mathcal{Y}\, {\rm div} \mathcal{Z}- \langle {}^S\nabla \mathcal{Y}, \nabla \mathcal{Z}\rangle - \frac{1}{2\varepsilon-R} \langle \pi, \nabla \mathcal{Y}\rangle \langle \pi, \nabla\mathcal{Z}\rangle\bigg)\, {\rm d}M.
\end{align*}
where $\nu$ stands for the outward normal to $\partial M$. Since the integrand of the first term in the right-hand side vanishes at the boundary, it follows that, if $L\mathcal{Y}=0$, then
\begin{equation}
\int_\Omega \frac{1}{N}\bigg({\rm div}\mathcal{Y}\, {\rm div} \mathcal{Z}- \langle {}^S\nabla \mathcal{Y}, \nabla \mathcal{Z}\rangle - \frac{1}{2\varepsilon-R} \langle \pi, \nabla \mathcal{Y}\rangle \langle \pi, \nabla\mathcal{Z}\rangle\bigg) {\rm d}M =0.
\end{equation}
In particular on $\Omega' \subset \Omega$ we have
\begin{equation}
\label{euler-lagrange1}
\int_{\Omega'} \frac{1}{N}\bigg(\langle {}^S\nabla \mathcal{Y}, \nabla \mathcal{Y}\rangle - ({\rm div}\mathcal{Y})^2+   \frac{1}{2\varepsilon-R} \langle \pi, \nabla \mathcal{Y}\rangle^2\bigg){\rm d}M=0.
\end{equation}
Now define
\begin{align}\label{el2}
I\doteq \langle {}^S\nabla \mathcal{Y}, \nabla \mathcal{Y}\rangle - ({\rm div}\mathcal{Y})^2+   \frac{1}{2\varepsilon-R} \langle \pi, \nabla \mathcal{Y}\rangle^2.
\end{align}
Since $(\psi,\beta)$ gives a reference solution for the constraint equations, using (\ref{hamit}) we know that
\begin{align*}
\frac{({\rm tr} K)^2-\vert K\vert_{g}^2}{2\epsilon-R}=1.
\end{align*}
We note that
\[
\langle {}^S\nabla \mathcal{Y}, \nabla \mathcal{Y}\rangle = \langle {}^S\nabla \mathcal{Y}, {}^S\nabla \mathcal{Y}\rangle = |{}^S \nabla \mathcal{Y}|_{ g}^2
\]
and
\[
\langle \pi, \nabla \mathcal{Y}\rangle = \langle \pi, {}^S\nabla\mathcal{Y}\rangle.
\]
We also have
\[
{\rm div}\mathcal{Y} =  g^{ij}\nabla_i \mathcal{Y}_j = \langle g, \nabla \mathcal{Y}\rangle.
\]
and
\[
\langle K, \nabla \mathcal{Y}\rangle = \langle K, {}^S\nabla\mathcal{Y}\rangle.
\]
Note that since $\pi$ is definite, then $\vert K\vert_{g}\neq 0$. Thus we denote
\[
\lambda = \frac{1}{|K|_{g}^2}\langle \nabla\mathcal{Y}, K\rangle
\]
and rewrite (\ref{el2}) above as
\begin{align*}
I &= |{}^S \nabla \mathcal{Y} - \lambda K|_{ g}^2 +2\lambda \langle {}^S\nabla\mathcal{Y}, K\rangle -\lambda^2 |K|_{ g}^2  - ({\rm div}\mathcal{Y})^2+   \frac{1}{2\varepsilon-R} \langle \pi, \nabla \mathcal{Y}\rangle^2\\
& = |{}^S \nabla \mathcal{Y} - \lambda K|_{ g}^2 +2\frac{1}{|K|_{ g}^2} \langle {}^S\nabla\mathcal{Y}, K\rangle^2 -\frac{1} {|K|_{ g}^2}\langle {}^S\nabla\mathcal{Y}, K\rangle^2  - ({\rm div}\mathcal{Y})^2+   \frac{1}{2\varepsilon-R} \langle \pi, \nabla \mathcal{Y}\rangle^2.
\end{align*}
However, since
\[
\pi = K- {\rm tr} K\, g
\]
and
\begin{eqnarray*}
& & \langle \pi, \nabla \mathcal{Y}\rangle = \langle K, \nabla \mathcal{Y}\rangle - {\rm tr} K\langle g, \nabla \mathcal{Y}\rangle =\langle {}^S\nabla\mathcal{Y}, K\rangle - {\rm tr}K\, {\rm div}\mathcal{Y},
\end{eqnarray*}
we get that
\begin{align*}
I  &= |{}^S \nabla \mathcal{Y} - \lambda K|_{ g}^2 +\frac{1}{|K|_{ g}^2} \langle {}^S\nabla\mathcal{Y}, K\rangle^2   - ({\rm div}\mathcal{Y})^2\\
& +   \frac{1}{2\varepsilon-R} \big(\langle {}^S\nabla\mathcal{Y}, K\rangle^2 -2{\rm tr}K\, {\rm div}\mathcal{Y} \langle {}^S\nabla\mathcal{Y}, K\rangle  +({\rm tr}K)^2 ({\rm div}\mathcal{Y})^2\big).
\end{align*}
Using again the fact that
\[
2\varepsilon-R = ({\rm tr} K)^2-|K|_{g}^2
\]
we have
\begin{eqnarray*}
& &  I  = |{}^S \nabla \mathcal{Y} - \lambda K|_{ g}^2 +\bigg(\frac{1}{|K|_{ g}^2}+\frac{1}{ ({\rm tr} K)^2-|K|_{g}^2}\bigg) \langle {}^S\nabla\mathcal{Y}, K\rangle^2 + \bigg( \frac{({\rm tr}K)^2}{({\rm tr} K)^2-|K|_{g}^2}  -1\bigg) ({\rm div}\mathcal{Y})^2\\
& & \,\,-   \frac{2}{2\varepsilon-R} {\rm tr}K\, {\rm div}\mathcal{Y} \langle {}^S\nabla\mathcal{Y}, K\rangle.
\end{eqnarray*}
Therefore
\begin{eqnarray*}
& &  I  = |{}^S \nabla \mathcal{Y} - \lambda K|_{ g}^2 +\frac{({\rm tr} K)^2}{|K|^2( ({\rm tr} K)^2-|K|_{g}^2)} \langle {}^S\nabla\mathcal{Y}, K\rangle^2 + \frac{ |K|_{g}^2}{({\rm tr} K)^2-|K|_{g}^2}   ({\rm div}\mathcal{Y})^2\\
& & \,\,-   \frac{2}{({\rm tr} K)^2-|K|_{g}^2} {\rm tr}K\, {\rm div}\mathcal{Y} \langle {}^S\nabla\mathcal{Y}, K\rangle.
\end{eqnarray*}
Hence we have
\begin{align*}
(({\rm tr} K)^2-|K|_{g}^2)\, I &= (({\rm tr} K)^2-|K|_{g}^2) |{}^S \nabla \mathcal{Y} - \lambda K|_{ g}^2 +\frac{({\rm tr} K)^2}{|K|_{g}^2} \langle {}^S\nabla\mathcal{Y}, K\rangle^2 +  |K|_{g}^2  ({\rm div}\mathcal{Y})^2\\
& -   2 {\rm tr}K\, {\rm div}\mathcal{Y} \langle {}^S\nabla\mathcal{Y}, K\rangle\\
&  = (({\rm tr} K)^2-|K|_{g}^2) |{}^S \nabla \mathcal{Y} - \lambda K|_{ g}^2 +\frac{({\rm tr} K)^2}{|K|_{g}^2} \langle {}^S\nabla\mathcal{Y}, K\rangle^2 +  |K|_{g}^2  ({\rm div}\mathcal{Y})^2\\
& -   2\frac{ {\rm tr}K}{|K|_{g}} |K|_{g} {\rm div}\mathcal{Y} \langle {}^S\nabla\mathcal{Y}, K\rangle.
\end{align*}
Therefore
\begin{eqnarray*}
& & (({\rm tr} K)^2-|K|^2)\, I  = (({\rm tr} K)^2-|K|^2) \bigg|{}^S \nabla \mathcal{Y} - \frac{1}{|K|_{g}^2}\langle \nabla\mathcal{Y}, K\rangle K\bigg|_{ g}^2 +\bigg(|K|_{g} {\rm div}\mathcal{Y}-\frac{{\rm tr} K}{|K|_{g}} \langle {}^S\nabla\mathcal{Y}, K\rangle\bigg)^2.
\end{eqnarray*}	
Using this in (\ref{euler-lagrange1}), we get the following
\begin{align}\label{injectivity}
\int_{\Omega'}\frac{1}{N}|{}^S \nabla \mathcal{Y} - \frac{1}{|K|_{g}^2}\langle \nabla\mathcal{Y}, K\rangle K|_{ g}^2{\rm d}M_g+\int_{\Omega'}\frac{1}{(2\epsilon-R)N}\bigg(|K|_{g} {\rm div}\mathcal{Y}-\frac{{\rm tr} K}{|K|_{g}} \langle {}^S\nabla\mathcal{Y}, K\rangle\bigg)^2{\rm d}M_g=0.
\end{align}
Since, by hypothesis, $2\epsilon-\overline{R}>0$, then both integrands are non-negative, thus, for the equality to hold, both must equal zero. From the first integral, we get that
\begin{equation}\label{inj8}
{}^S \nabla \mathcal{Y} = \frac{1}{|K|_g^2}\langle \nabla\mathcal{Y}, K\rangle K
%{}^{S}\overline{\nabla}y^b=\frac{K\cdot \overline{\nabla}y}{K^2}K.
\end{equation}
Taking traces we get
\[
{\rm div}\mathcal{Y} = \frac{1}{|K|_{g}^2}\langle \nabla\mathcal{Y}, K\rangle {\rm tr} K
\]
and multiplying by $|K|_g$ we obtain
\[
|K|_g\, {\rm div}\mathcal{Y} = \frac{ {\rm tr} K}{|K|_g}\langle \nabla\mathcal{Y}, K\rangle.
\]
which proves that (\ref{injectivity}) holds if and only if (\ref{inj8}) holds. This shows that if $\mathcal{Y}\in \ker L$ then $\mathcal{Y}$ satisfies (\ref{inj8}). The converse for this statement is also true. If we consider a field $\mathcal{Y}$ which satisfies (\ref{inj8}), after some computations, we get that $L\mathcal{Y}=0$. So if (\ref{inj8}) has only the trivial solution $\mathcal{Y}=0$, then $L$ is injective.
\end{proof}

At this point it is interesting to note that the curvature condition
$2\epsilon-R_{g}>0$ by itself does not pose any topological obstruction. This
is because, for instance, if $\epsilon$ is continuous, then\ the compactness
of $M$ implies that $\epsilon$ is bounded, hence any metric on $M$ with scalar
curvature which is \textit{more negative} than $2\min_{p\in M}\epsilon(p)$
satisfies this condition. That this last (stronger) condition can always be
satisfied can be seen as a consequence of an important result in geometric
analysis, which shows that on any compact $n$-dimensional manifold, $n\geq3$,
we can always choose a smooth metric $g^{\prime}$ such that $R(g^{\prime})=-1$
(see \cite{Aubin2},\cite{Lohkamp}). We then can always find a suitably
rescaled metric $g$ satisfying $2\epsilon-R_{g}>0$. Later on, using more
subtle arguments, we will actually show that one  can always find such metric
within an initial data set satisfying the constraint equations. Note that,
under our assumption, the Sobolev embedding theorems imply that we are
assuming $\epsilon$ at least $C^{1}$, and thus, under our regularity
hypotheses, this inequality does not impose any \textit{a priori} restriction.
\bigskip

Using these results and applying the implicit function theorem we have the
following theorem.

\begin{thm}
Suppose $(\psi_{0},\beta_{0})\in\mathcal{E}_{1}\times\mathcal{E}_{2}$
satisfies $\Phi(\psi_{0},\beta_{0})=0$. Then if $\pi$ is a definite operator
at each point of $M$, $2\epsilon-R_{g}>0$ everywhere on $M$, and if for a
given function $\mu$ on $M$ the equation
\[
{}^{S}\nabla\mathcal{Y}=\mu K
\]
has only the solution $\mathcal{Y}=0$, $\mu=0$, then there are open
neighbourhoods $\mathcal{V}\subset\mathcal{E}_{1}$ and $\mathcal{W}\subset\mathcal{E}_{2}$ of
$\psi_{0}$ and $\beta_{0}$ respectively, and a unique mapping
\[
g:\mathcal{V}\to\mathcal{W}
\]
such that $\Phi(\psi,g(\psi))=0$ for all $\psi\in\mathcal{V}$.
\end{thm}

Notice that this theorem shows that given an initial data set $\psi_{0}%
\in\mathcal{E}_{1}$ for which a solution $\beta_{0}$ of the reduced constraint
equations $\Phi(\psi,\beta)=0$ exists, if the conditions stated in the theorem
are satisfied, then for every $\psi\in\mathcal{V}\subset\mathcal{E}_{1}$ there
is a unique solution of the reduced constraint equations. Then taking lapse
defined as in (\ref{lapse}), we get a solution of the full constraint
equations. This answers our original question.

At this point, we would like to show that the conditions stated in the
previous theorem are not too restrictive. By this we mean that, generically,
there are solutions of the constraint equations satisfying all these
conditions. With this in mind, notice that, given a solution $(g,K)$ for
(\ref{hamit})-(\ref{momentum}) satisfying all the hypotheses of the previous
theorem, we can use this solution obtained in the usual way, to obtain a
reference solution for the TSP. In order to do this, just consider any given
$N\in H_{s}$, $N>0$, and $\beta\in\mathcal{E}_{2}$ and take
\begin{equation}
\dot g_{ij}\doteq2NK_{ij}+(\nabla_{i} \beta_{j}+\nabla_{j} \beta_{i}).
\end{equation}
Then the set $((g, \dot g,\epsilon,S), \beta)$ gives a reference solution of
the constraint equations. Thus, what we need to show is that the constraint
equations (\ref{hamit})-(\ref{momentum}) on a compact manifold $M$ always
admit a solution $(g,K)$, satisfying all the hypotheses of the theorem. A
first step in this direction is the following proposition.

\begin{prop}
Suppose $(M, g)$ is a an n-dimensional compact Riemannian manifold. Suppose
that $(g,K)$ satisfy the constraint equations \textrm{(\ref{hamit})} and
\textrm{(\ref{momentum})}, where $K$ is a $(0,2)$ tensor field, and suppose
that $2\epsilon-R_{g}>0$ on M. Then, if the Ricci tensor on $M$ is negative
definite, the equation ${}^{S}\nabla\mathcal{Y}=\mu K$ has only the trivial
solution $\mathcal{Y}=0$ and $\mu=0$.
\end{prop}

\begin{proof}
Suppose $\mathcal{Y}$ and $\mu$ satisfy ${}^{S}\nabla\mathcal{Y}=\mu K$. From the definition of the curvature tensor we get the following
\begin{align*}
{R}_{ij}\mathcal{Y}^j&= \nabla_j\nabla_i \mathcal{Y}^j-\nabla_i\nabla_j \mathcal{Y}^j\\
&= 2\nabla_j {}^S \nabla_i \mathcal{Y}^j-\nabla_j \nabla^j \mathcal{Y}_i-\nabla_i(\mu K^j_j)\\
&= 2\nabla_j (\mu K_i^j)-\nabla_j\nabla^j \mathcal{Y}_i-\nabla_i(\mu K^j_j).
\end{align*}
Then we have
\begin{align*}
R_{ij}\mathcal{Y}^j\mathcal{Y}^i =2\mathcal{Y}^i\nabla_j (\mu K_i^j)-\mathcal{Y}^i\nabla_j \nabla^j \mathcal{Y}_i-\mathcal{Y}^i\nabla_i(\mu K^j_j)
\end{align*}
We can write this last expression in the following way, which is globally defined:
\begin{equation}
{\rm Ric}_{g }(\mathcal{Y}, \mathcal{Y})=2\, {\rm div}(\mu K)\cdot \mathcal{Y} -\langle \mathcal{Y}, \Delta \mathcal{Y}\rangle-\langle \nabla{\rm tr}_{g}\,\mu K, \mathcal{Y}\rangle.
\end{equation}
In this last expression, $\Delta$ stands for the connection Laplacian defined by ${\rm tr}_g \nabla^2$.  From the previous expression we get that
\begin{equation}\label{negricc1}
\int_M\big( {\rm Ric}_{ g }(\mathcal{Y}, \mathcal{Y})+\langle \mathcal{Y}, \Delta \mathcal{Y}\rangle+\langle \nabla(\mu{\rm tr}_{g}K), \mathcal{Y}\rangle-2\, {\rm div}(\mu K)\cdot \mathcal{Y}\big)\, {\rm d}M_g=0
\end{equation}
Applying divergence theorem, we get
\begin{equation}\label{negricc2}
\int_M ( {\rm Ric}_{ g }(\mathcal{Y}, \mathcal{Y})-|\nabla\mathcal{Y}|^2_g-\mu^2 (({\rm tr}_g K)^2 - |K|^2_g)+ \mu \langle K, \nabla \mathcal{Y}\rangle\big)\, {\rm d}M_g =0
%\Big( \overline{R}icc(y,y)-\vert \overline{\nabla}y\vert^2_{\overline{g}}-\mu^2((tr_{\overline{g}}K)^2-\vert K\vert^2_{\overline{g}})+\mu K\cdot\overline{\nabla}y^b\Big)\mu_{\overline{g}}.
\end{equation}
Now we will analyze the second and fourth terms in the integrand. In order to do this, we rewrite these expressions in the following way:
\begin{align*}
\vert \nabla\mathcal{Y}\vert^2_{g}-\langle \mu  K, \nabla\mathcal{Y}\rangle &= \langle \nabla\mathcal{Y}, \nabla \mathcal{Y}\rangle - \langle {}^S\nabla\mathcal{Y}, \nabla \mathcal{Y}\rangle =\langle \nabla\mathcal{Y}-{}^S\nabla\mathcal{Y}, \nabla\mathcal{Y}\rangle\\
&=\langle {}^A\nabla \mathcal{Y}, \nabla\mathcal{Y}\rangle = \langle {}^A\nabla \mathcal{Y}, {}^A\nabla\mathcal{Y}\rangle \\
&= |{}^A\nabla \mathcal{Y}|^2_g,
\end{align*}
where
\[
{}^A\nabla_i\mathcal{Y}_j = \nabla_i \mathcal{Y}_j - {}^S\nabla_i\mathcal{Y}_j = \frac{1}{2}\big(\nabla_i\mathcal{Y}_j-\nabla_j\mathcal{Y}_i\big).
\]
Then, (\ref{negricc2}) can be rewritten as
\begin{equation}
\int_M ( {\rm Ric}_{g}(\mathcal{Y}, \mathcal{Y})-|{}^A\nabla\mathcal{Y}|^2_g-\mu^2 ( 2\epsilon- R_g)\big)\, {\rm d}M_g=0
%\Big( \overline{R}icc(y,y)-\vert {}^{A}\overline{\nabla}y\vert^2_{\overline{g}}-\mu^2(2\epsilon-\overline{R}) \Big)\mu_{\overline{g}} =0
\end{equation}
Thus, if ${\rm Ric}_{g}$ is negative definite, then the integrand is non-positive. Hence in order for the last equality to hold, each term in the integrand has to equal zero. So the only possible $\mathcal{Y}$ and $\mu$ which can satisfy ${}^{S}\nabla \mathcal{Y}= \mu K$ under these geometric conditions are $\mathcal{Y}=0$, $\mu=0$. \end{proof}

This proposition implies that, given a solution of the constraint equations
$(g_{0},K_{0})$ satisfying $2\epsilon-R>0$, then, if $\pi$ is a definite
operator on $M$, and $\mathrm{Ric}_{g_{0}}$ is negative definite, then the
linearization $L=\frac{\delta\Phi}{\delta\beta}\big|_{(\psi_{0},\beta_{0})}$
is an isomorphism and Theorem 1 applies. A 3-dimensional version of the
previous proposition was shown in \cite{Bartnik}.

We will now show that any compact $n$-dimensional manifold admits a solution
of the constraint equations satisfying all the hypotheses of Theorem 1. The
first step in this direction is to look for a solution of the constraint
equations of the form $(h,\alpha h)$, with $h$ being a Riemannian metric and
$\alpha$ a positive constant. That is, we are considering $K=\alpha h$ from
the beginning. We will restrict ourselves to solutions of (\ref{hamit}%
)-(\ref{momentum}) with $S=0$, i.e, with zero momentum density. With this set
up, equation (\ref{momentum}) is automatically satisfied and we just need to
consider equation (\ref{hamit}), which, under these conditions, becomes the
following equation, which is posed for $h$:
\begin{align}
\label{giulini1}R_{h}=2\epsilon-\alpha^{2}n(n-1).
\end{align}

In order to guarantee the existence of solutions for (\ref{giulini1}), we will
appeal to the following well-established theorem:

\begin{thm}
Let $M$ be a $C^{\infty}$ compact manifold of dimension $n\geq3$. If $f\in
C^{\infty}(M)$ is negative somewhere, then there is a $C^{\infty}$ Riemannian
metric on $M$ with $f$ as its scalar curvature.
\end{thm}

This theorem was proved by Kazdan and Warner \cite{KW}, and its proof can also
be found in \cite{Aubin}. Using this theorem, we see that, if the right-hand
side of (\ref{giulini1}) is negative somewhere, then, for smooth sources
$(\epsilon\in C^{\infty})$, we have that (\ref{giulini1}) always admits a
smooth solution. In order to guarantee this last condition, just take
$\alpha^2>\min{\frac{2\epsilon}{n(n-1)}}$. A solution constructed in this way
satisfies two of the three conditions required by Theorem 1, that is, it satisfies

\begin{itemize}
\item $2\epsilon-R_{h}>0$, which comes from (\ref{giulini1}).

\item $\pi$ is negative definite, since from $K=\alpha h$ we get that
$\pi=\alpha(1-n)h$.
\end{itemize}

In this context, the last condition of Theorem 1 becomes the statement that
$h$ does not admit conformal Killing fields. We will show that we can always
find a solution $h$ of (\ref{giulini1}) with this property. In order to do
this, we need to make a remark on how Theorem 2 is proved (See, for example,
\cite{Aubin} chapter 6). The proof of this theorem begins with the statement
that we can choose on $M$ a Riemannian metric $g^{\prime}$, with
$R_{g^{\prime}}=-1$, which is something known from \cite{Aubin2}, and then one
finds a conformal metric to $g^{\prime}$ satisfying the theorem. In fact, it
is shown that $h$ has the following form:
\begin{align}
\label{generic1}h=(\phi^{-1})^{*}(u^{\frac{4}{n-2}}g^{\prime})
\end{align}
where $u$ is a positive function and $\phi$ is a suitably chosen
diffeomorphism. In this process, we claim that we can choose $g^{\prime}$
without conformal Killing fields. We support this claim using the results
shown by Lohkamp in \cite{Lohkamp}. There, it is shown that every manifold $M$
of dimension $n\geq3$ admits a complete metric with negative definite Ricci
tensor. As a corollary of this theorem, it is shown that, starting from such a
Riemannian metric $g$ on $M$ with negative definite Ricci tensor, we can find
a conformal metric $g^{\prime}=v^{\frac{4}{n-2}}g$, such that $R(g^{\prime})=-1$.
In this way, using this metric as the starting point in the proof of Theorem
2, we get that the metric $h$ solving (\ref{giulini1}) has the following
form:
\begin{align}
\label{generic2}h=(\phi^{-1})^{*}(u^{\frac{4}{n-2}}v^{\frac{4}{n-2}}g)
\end{align}
where $g$ has negative definite Ricci tensor. We now have the following:

\begin{prop}
The solution \textrm{(\ref{generic2})} obtained for \textrm{(\ref{giulini1})}
does not admit conformal Killing fields.
\end{prop}

\begin{proof}
It is a well-known fact that, on compact manifolds, metrics with negative definite Ricci tensor do not admit conformal Killing fields (see, for instance, \cite{C-B1} chapter 7), which is the case of the metric $g$. Now, imagine that $h$ admitted a conformal Killing field $Y\in\Gamma(TM)$, that is $\pounds_Yh=\lambda h$, for some $\lambda\in C^\infty(M)$, where $\pounds_Yh$ is the Lie derivative of $h$ with respect to $Y$. From (\ref{generic2}) we get that
\begin{align*}
g=(uv)^{\frac{-4}{n-2}}\phi^{*}(h)\doteq \mu\phi^{*}(h).
\end{align*}
Define $X\doteq \phi^{-1}_{*}Y\in\Gamma(TM)$, so that $Y=\phi_{*}X$. We claim that under these conditions $X$ is a conformal Killing field of $g$. To check this, we have to compute the Lie derivative of $g$ with respect to $X$, that is
\begin{align*}
\pounds_Xg=X(\mu)\phi^{*}(h)+\mu\pounds_X(\phi^{*}h).
\end{align*}
Using that $\pounds_X(\phi^{*}h)=\phi^{*}(\pounds_Yh)$ (see \cite{AMR}) and the fact that $Y$ is a conformal Killing field for $h$, then we get that
\begin{align*}
\pounds_Xg=\big(X(\log\mu)+\lambda\circ\phi\big)g
\end{align*}
which shows that $X$ is a conformal Killing field for $g$. But this contradicts the fact that $g$ has negative definite Ricci tensor, thus $h$ can not admit conformal Killing fields.
\end{proof}

Thus, we conclude that the solution we have constructed for the constraint
equations satisfies all the hypotheses of Theorem 1 and therefore can be used
as a reference solution. Then we can state the following theorem.

\begin{thm}
\label{thmTS2} Any smooth compact $n$-dimensional manifold $M$, $n\geq3$,
admits a smooth solution $(N,\beta)$ of the constraint equations
(\ref{hamit})-(\ref{momentum}) with $S=0$, with prescribed smooth free data
$\psi=(g,\dot{g},\epsilon_{\alpha},0)\in\mathcal{E}_{1}$, such that in an
$\mathcal{E}_{1}$-neighbourhood of $\psi$ the TSP is well-posed.
\end{thm}

\section{Final Remarks}

In this paper we have been able to show the validity of the main results
presented in \cite{Bartnik} in arbitrary dimensions ($n\geq3$). Specifically
we have shown that Wheeler's thin sandwich conjecture is true under certain
geometrical conditions in all these cases. As a novelty, we have also been
able to establish that the geometric hypotheses needed to prove this result
can always be satisfied in the case of zero momentum density, and thus that,
in these cases, there is an open subset in the space of possible initial data
for the constraint equations where the thin sandwich problem is well-posed.
These are interesting results describing the space of solutions of the
constraint equations in arbitrary dimensions. These type of results are
relevant in the study of the Cauchy problem for higher dimensional theories of
gravity, and they also give us a better understanding of the Superspace
picture for space-time in this context. In contrast to the usual approach to
the constraint equations, where the structure of the space of solutions and
its relation to properties of the evolving space-time is something we
understand quite well (see, for instance, \cite{C-B1},\cite{FMM}%
,\cite{Chrusciel}), most of these problems require further study in the
context of the Superspace picture for space-time.

Finally, it is worth to point out that even though we have not followed the
strategy presented in \cite{Giulini2} regarding the study of the TSP, the
results obtained therein suggest further research in the investigation of the
TSP in, perhaps, more physically realistic situations, since in this work
matter fields are included into the picture. Moreover, the results obtained
in  \cite{Giulini2}  offer some techniques which could complement the ones
presented here, such as the global uniqueness theorem which is presented
there (following the lines of \cite{B-O}), and a method to construct families
of reference solutions once one solutions satisfying the geometric restrictions
needed is obtained. Thus, we regard the combination of the two strategies as a
starting point for future research in this area.

\section*{Acknowledgements}

\noindent R. A. and C. R. would like to thank CNPq and CLAF for financial
support. J. L would like to thank CNPq and FUNCAP/CNPq/PRONEX for financial
support. We thank the referee for valuable comments and suggestions.

\end{document}